\documentclass[11pt]{article}

\usepackage{amsmath,amsthm,amssymb,amsfonts}

\usepackage{geometry} 
\newtheorem{theorem}{Theorem}

\newtheorem{example}{Example}

\newtheorem{lemma}{Lemma}
\newtheorem{corollary}{Corollary}

\begin{document}

\title{Quantum Uncertainty Based on Metric Adjusted Skew Information}

\author{Liang Cai\\
      School of Mathematics and Statistics,\\
      Beijing Institute of Technology, Beijing 100081, China\\
      cailiang@bit.edu.cn
       }

\date{}

\maketitle

\begin{abstract}

Prompted by the open questions in Gibilisco [Int. J. Software Informatics, 8(3-4): 265, 2014], in which he introduced a family of measurement-induced quantum uncertainty measures via metric adjusted skew informations, we investigate these measures' fundamental properties (including basis independence and spectral representation), and illustrate their applications to detect quantum nonlocality and entanglement.  
\\[1ex]
\\[1ex]

\end{abstract}


\section{Introduction}
Usually quantum uncertainty of a quantum state $\rho$ is measured by the von Neumann entropy
\begin{eqnarray}
S(\rho)=-\mathrm{tr}\rho \log \rho. \notag
\end{eqnarray}
It is a natural quantum counterpart of Shannon's entropy and plays an important role in quantum information theory (see for example \cite{MR1796805}). On the other hand it coincides with Shannon's entropy only when the measurement basis coincides with the eigenbasis of the density matrix $\rho$. So it may not capture all informational regularities and peculiarities of quantum states. If we evaluate Shannon's entropy under two or more ``mutually unbiased bases'', then the entropic uncertainty relation is developed and applied to quantum communication (see for example \cite{MR8426024}). One operationally invariant information measure is proposed by Brukner and Zeilinger$^{\cite{MR1720173}}$ 
\begin{eqnarray}
\tilde{S}(\rho)=\mathrm{tr} \rho^2-\frac{1}{n}, \notag
\end{eqnarray}
where $n$ is the dimension of the quantum system.
Luo$^{\cite{MR2338161}}$ interpreted this measure as the total variance of an observable basis  $\{H_j:j=1,...,n^2\}$ under a state $\rho$
\begin{eqnarray}
U(\rho)=\sum_{j=1}^{n^2}V(\rho,H_j)=n-\mathrm{tr} \rho^2, \notag
\end{eqnarray}
which is independent of the choice of the basis. Here $V(\rho,H)$ denotes the variance of an observable $H$ 
\begin{eqnarray}
V(\rho,H)=\mathrm{tr}\rho H^2 -(\mathrm{tr}\rho H)^2.\notag
\end{eqnarray}
And the basis $\{H_j:j=1,...,n^2\}$ is an orthonormal basis of the real Hilbert space of Hermitian operators on the quantum system with inner product $\langle A,B\rangle=\mathrm{tr} AB$.

In one earlier article, Luo$^{\cite{PhysRevA.73.022324}}$ replaced the variance by the Wigner-Yanase skew information 
\begin{eqnarray}
I^{WY}(\rho,H)=-\frac{1}{2}\mathrm{tr}[\sqrt{\rho},H]^2,\notag
\end{eqnarray}
and derived a measure of quantum uncertainty 
\begin{eqnarray}\label{Q of WY}
Q^{WY}(\rho)=\sum_{j=1}^{n^2}I^{WY}(\rho,H_j)=n-(\mathrm{tr}\sqrt{\rho})^2.
\end{eqnarray}
Mathematically, it is connected with Tsallis entropy $S_q(\rho):=(1-\mathrm{tr}\rho^q)/(q-1)$ with index $q=1/2$ (see \cite{MR968597}). It is basis independent and convex with respect to $\rho$. Further, this measure can be applied to detect quantum nonlocality and entanglement (see \cite{PhysRevA.88.014301,PhysRevA.85.032117}), so that this measure can be regarded as an important index of quantum correlations. It is well known the Wigner-Yanase information is a special type of the quantum Fisher information$^{\cite{MR2003930}}$
(or the metric adjusted skew information$^{\cite{MR2430207}}$). Gibilisco (section 5 in \cite{gibiliscofisherinf}) proposed a general form of the quantum uncertainty 
\begin{eqnarray}\label{main quantity}
Q^{f}(\rho)=\sum_{j=1}^{n^2}I^{f}(\rho,H_j),
\end{eqnarray}
and suggested to prove that it is well defined and investigate its applications. We denote by $\mathcal{F}_{op}$ the set of all functions $f:R_+ \rightarrow R_+$ such that
\begin{enumerate}
\item[(i)] $f$ is operator monotone,
\item[(ii)] $f(t)=tf(t^{-1})$ for all $t>0$,
\item[(iii)] f(1)=1.
\end{enumerate}
We say that $f$ is regular if $f(0)>0$ and non-regular if $f(0)=0$. In the sequel, $f$ always denotes a regular function in $\mathcal{F}_{op}$. For a regular $f\in \mathcal{F}_{op}$, $I^f(\rho,H)$ denotes the metric adjusted skew information
\begin{eqnarray}\label{quantum fisher in}
I^f(\rho,H)=\frac{f(0)}{2}\mathrm{tr}\,i[\rho,H]\frac{1}{m^f(L_\rho, R_\rho)}i[\rho,H],
\end{eqnarray}
adjusted by an operator mean 
\begin{eqnarray}
m^f(x,y)=xf(\frac{y}{x}).\notag
\end{eqnarray}
$L_\rho$ and $R_\rho$ are the positive definite commuting left and right multiplication operators by $\rho$. The metric adjusted skew information has the following interesting properties (see \cite{MR2661519,MR2003930,MR2430207}).

(a) $I^f(\rho,H)=V(\rho,H)$, when $\rho$ is pure. For mixed state $\rho$, we have 
\begin{eqnarray}
I^f(\rho,H)\leq V(\rho,H).\notag
\end{eqnarray}

(b) $I^f(\rho,H)=0$ whenever $\rho$ commutes with $H$, and $I^f(U\rho U^{\dag},UHU^{\dag})=I^f(\rho,H)$ for any unitary operator $U$. 

(c) The metric adjusted skew information decreases when several states are mixed:
\begin{eqnarray}
I^f(\sum_j \alpha_j \rho_j,H)\leq \sum_j \alpha_j I(\rho_j,H),\notag
\end{eqnarray}
where $\sum_j \alpha_j=1,\alpha_j \geq 0$.

(d) The metric adjusted skew information satisfies the weak form of superadditivity. Let $\rho^{ab}$ be a state of composite system $H^a \otimes H^b$, then for an observable $A$ on $H^a$, we have
\begin{eqnarray}\label{weak superadditivity}
I^f(\rho^{ab}, A \otimes \mathbf{1}^b)\geq I^f(\rho^a, A).
\end{eqnarray}

Our main result is to show that all the quantum uncertainty meaures defined by (\ref{main quantity}) have common essential properties, which will be proved and summarized in Section \ref{secofproperties}, and all the quantum uncertainty measures defined by (\ref{main quantity}) can be applied to detect the quantum nonlocality and entanglement, which will be investigated in section \ref{secofapplications}. Finally, we will give more discussions in Section \ref{discussion}.


\section{Basis independence and spectral representation}\label{secofproperties}

We define a monotone metric on the observable space by (see \cite{MR1403277}) 
\begin{eqnarray}
K_\rho^f(A,B)=\frac{f(0)}{2}\mathrm{tr} A\frac{1}{m^f(L_\rho,R_\rho)}B,\notag
\end{eqnarray}
whose monotonicity means that it decreases under a quantum operation. That is to say, for any linear completely positive map $T$ with trace preserving property, 
\begin{eqnarray}
K_{T(\rho)}^f(T(A),T(B))\leq K_\rho^f(A,B).\notag
\end{eqnarray}
Therefore the metric adjusted skew information (\ref{quantum fisher in}) can be rewritten as
\begin{eqnarray}
I^f(\rho,H)=K_\rho^f(i[\rho,H],i[\rho,H]).\notag
\end{eqnarray} 
Obviously, $K_\rho^f(A,B)$ is bilinear on $\mathbb{R}$. 

Consider an $n$-dimensional quantum system $H$ with system Hilbert space $\mathbb{C}^n$. The set of all observables on $H$ (i.e. self-adjoint operators on $H$) constitutes a real $n^2$-dimensional Hilbert space $L(H)$ with inner product $\langle A,B\rangle=\mathrm{tr} AB$. We will show that the quantity (\ref{main quantity}) is basis independent.

\begin{theorem}\label{theoremofbasisindependent}
Let $\{H_j\}$ and $\{K_j\}$ be two different orthonormal bases of $L(H)$. Then
\begin{eqnarray}
\sum_{j=1}^{n^2}I^{f}(\rho,H_j)=\sum_{j=1}^{n^2}I^{f}(\rho,K_j).\notag
\end{eqnarray}
\end{theorem}

\begin{proof}
Our proof heavily depends on the bilinearity of $K_\rho^f(A,B)$.  We may write 
\begin{eqnarray}
K_i=\sum_{j=1}^{n^2}a_{ij}H_j, \quad i=1,2,...,n^2\notag
\end{eqnarray}
with $\{a_{ij}\}_{n^2\times n^2}$ a real orthogonal matrix, hence we have
\begin{eqnarray}
\sum_{i=1}^{n^2} a_{ij}a_{ik}=\delta_{jk}, \quad j,k=1,2,...,n^2.\notag
\end{eqnarray}
Consequently,
\begin{eqnarray}
&&\sum_{j=1}^{n^2}I^{f}(\rho,K_j)=\sum_{j=1}^{n^2}\sum_{k,l=1}^{n^2}a_{jk}a_{jl}K_\rho^f(i[\rho,H_k],i[\rho,H_l])
\notag\\&&=\sum_{k,l=1}^{n^2}\Big(\sum_{j=1}^{n^2}a_{jk}a_{jl}\Big)K_\rho^f(i[\rho,H_k],i[\rho,H_l])
\notag\\&&=\sum_{j=1}^{n^2}K_\rho^f(i[\rho,H_j],i[\rho,H_j])=\sum_{j=1}^{n^2}I^{f}(\rho,H_j). \notag
\end{eqnarray}
\end{proof}

We have established the basis independence of $Q^f(\rho)$, so we can choose a special basis $\{H_j\}$ with matrix representation under the eigenbasis of $\rho$ to calculate the spectral representation of $Q^f(\rho)$.

\begin{theorem}\label{theoremofspectral}
Let $\{\lambda_j: j=1,...,n\}$ be the spectrum of $\rho$, then
\begin{eqnarray}\label{spectrum representation}
Q^f(\rho)=\frac{f(0)}{2}\sum_{k,l=1}^{n}\frac{(\lambda_k-\lambda_l)^2}{m^f(\lambda_k,\lambda_l)}.
\end{eqnarray}
\end{theorem}

\begin{proof}
Let $\{|\varphi_j\rangle: j=1,...,n\}$ be the eiginbasis corresponding to $\{\lambda_j: j=1,...,n\}$. We choose an orthonormal basis $\{H_j:j=1,...,n^2\}$ of $L(H)$ as
\begin{eqnarray}
\Big\{|\varphi_j\rangle \langle \varphi_j|: j=1,...,n\Big\}\bigcup \Big\{\frac{1}{\sqrt{2}}|\varphi_k\rangle \langle \varphi_j|+\frac{1}{\sqrt{2}}|\varphi_j\rangle \langle \varphi_k|: 1\leq k < j \leq n\Big\}
\notag\\\bigcup \Big\{\frac{i}{\sqrt{2}}|\varphi_k\rangle \langle \varphi_j|+\frac{-i}{\sqrt{2}}|\varphi_j\rangle \langle \varphi_k|: 1\leq k < j \leq n\Big\}.\notag
\end{eqnarray}
Then
\begin{eqnarray}
Q^f(\rho)&=&\sum_{j=1}^{n^2}I^f(\rho,H_j)   
\notag\\ &=&\sum_{j=1}^{n^2}\frac{f(0)}{2}\mathrm{tr} H_j\frac{(L_\rho-R_\rho)^2}{m^f(L_\rho,R_\rho)}H_j
\notag\\ &=&\sum_{j=1}^{n^2}\frac{f(0)}{2}\sum_{k,l=1}^{n}\frac{(\lambda_k-\lambda_l)^2}{m^f(\lambda_k,\lambda_l)}|\langle \varphi_k|H_j|\varphi_l\rangle|^2
\notag\\ &=& \frac{f(0)}{2}\sum_{k,l=1}^{n}\frac{(\lambda_k-\lambda_l)^2}{m^f(\lambda_k,\lambda_l)}. \notag
\end{eqnarray}
\end{proof}

The theorem in Luo \cite{PhysRevA.73.022324} (see (\ref{Q of WY})) can be recovered by taking 
\begin{eqnarray}\label{fofwy}
f(t)=f^{WY}(t):=\frac{(\sqrt{t}+1)^2}{4}.
\end{eqnarray}

These two theorems give positive answers to the conjecture i) and ii) at page 273 in \cite{gibiliscofisherinf}. For reader's convenience and more observations on $Q^f(\rho)$, we specify the conjecture ii) in \cite{gibiliscofisherinf} in detail. 

Denote the set of regular functions in $\mathcal{F}_{op}$ by $\mathcal{F}_{op}^r$, and the set of non-regular functions in $\mathcal{F}_{op}$ by $\mathcal{F}_{op}^n$. Set
\begin{eqnarray}
\tilde{f}:=\frac{1}{2}\Big[(x+1)-(x-1)^2\frac{f(0)}{f(x)}\Big].\notag
\end{eqnarray}
Then the correspondence $f\rightarrow \tilde{f}$ is a bijection between $\mathcal{F}_{op}^r$ and $\mathcal{F}_{op}^n$ (see \cite{MR2503967}). For example, let 
\begin{eqnarray}
f^{\mathrm{SLD}}(x)=\frac{x+1}{2},\notag
\end{eqnarray}
then its corresponding function in $\mathcal{F}_{op}^n$ is
\begin{eqnarray}
\tilde{f}^{\mathrm{SLD}}(x)=\frac{2x}{x+1}.\notag
\end{eqnarray}
They are the generators of the arithmetic mean and the harmonic mean respectively, i.e.
\begin{eqnarray}
m^{f^{\mathrm{SLD}}}(x,y)=\frac{x+y}{2}, \quad\quad m^{\tilde{f}^{\mathrm{SLD}}}(x,y)=\frac{2}{x^{-1}+y^{-1}}.\notag
\end{eqnarray}
Another interesting example of the $f\rightarrow \tilde{f}$ correspondence is that, if $0<\alpha <1$ then 
\begin{eqnarray}
f_{\alpha}(x)=\alpha(1-\alpha)\frac{(x-1)^2}{(x^{\alpha}-1)(x^{1-\alpha}-1)},\notag
\end{eqnarray}
corresponds to
\begin{eqnarray}
\tilde{f}_{\alpha}(x)=\frac{x^{\alpha}+x^{1-\alpha}}{2}.\notag
\end{eqnarray}

\begin{corollary}\label{corollary ftilda}
\begin{eqnarray}\label{spectral ftilda}
Q^f(\rho)=\sum_{k,l=1}^{n}\big[m_a(\lambda_k,\lambda_l)-m^{\tilde{f}}(\lambda_k,\lambda_l) \big],
\end{eqnarray}
where $m_a(\cdot,\cdot)$ is the arithmetic mean $m_a(x,y)=(x+y)/2$.
\end{corollary}

\begin{proof}
A direct algebraic calculation can show that
\begin{eqnarray}
m_a(x,y)-m^{\tilde{f}}(x,y)=\frac{f(0)}{2}\frac{(x-y)^2}{m^f(x,y)}. \notag
\end{eqnarray}
\end{proof}
Note that the two means $m_a(\cdot,\cdot)$ and $m^{\tilde{f}}(\cdot,\cdot)$ satisfy
\begin{enumerate}
\item[(i)] $m_a(x,x)=m^{\tilde{f}}(x,x)=x$,
\item[(ii)] $m_a(x,y)=m_a(y,x)$ and $m^{\tilde{f}}(x,y)=m^{\tilde{f}}(y,x)$.
\end{enumerate}
Thus the Corollary \ref{corollary ftilda} is exactly the conjecture ii) in \cite{gibiliscofisherinf}. And in the case
\begin{eqnarray}
\tilde{f}_{\alpha}(x)=\frac{x^{\alpha}+x^{1-\alpha}}{2},\notag
\end{eqnarray}
the authors obtained Formula (\ref{spectral ftilda}) at Page 149 of \cite{Li2011}. 

Due to the Corollary \ref{corollary ftilda}, we immediately get the following comparison result of $Q^f(\rho)$.
\begin{corollary}\label{corollary comparison}
If for any positive $x$ and $y$,
\begin{eqnarray}
m^{\tilde{f}}(x,y)\geq m^{\tilde{g}}(x,y),\notag
\end{eqnarray}
then we have
\begin{eqnarray}\label{comparison 1}
Q^f(\rho)\leq Q^g(\rho).
\end{eqnarray}
Particularly, for any $f\in \mathcal{F}^r_{op}$, we have
\begin{eqnarray}\label{comparison 2}
Q^f(\rho)\leq Q^{f^{\mathrm{SLD}}}(\rho).
\end{eqnarray}
\end{corollary}

\begin{proof}
It is well-known that among all the operator means, the harmonic one is the smallest one (see \cite{MR2439456,MR563399}), i.e., 
\begin{eqnarray}
m^{\tilde{f}}(y,x)\geq \frac{2}{x^{-1}+y^{-1}}=m^{\tilde{f}^{\mathrm SLD}}(x,y),\notag
\end{eqnarray}
for any positive $x$ and $y$. So in view of the comparison result (\ref{comparison 1}), we get (\ref{comparison 2}).
\end{proof}

\begin{corollary}\label{corollary n-1}
Let $\rho$ be a state of a quantum system with dimension $n$, then
\begin{eqnarray}\label{n-1}
0 \leq Q^f(\rho)\leq n-1,
\end{eqnarray}
and the bounds are tight, i.e., $ Q^f(\rho)=0$ when $\rho$ is $\mathbf{1}/n$ and $ Q^f(\rho)=n-1$ when $\rho$ is a pure state.
\end{corollary}

\begin{proof}
According to Theorem \ref{theoremofspectral}, we readily get $Q^f(\mathbf{1}/n)=0$. And 
if $\rho$ is a pure state denoted by $|\psi \rangle\langle \psi |$, we have
\begin{eqnarray}\label{tight bounds}
Q^f(|\psi \rangle\langle \psi |)=(n-1)f(0)\frac{(1-0)^2}{m^f(1,0)}=n-1.
\end{eqnarray}
Thanks to the convexity of the metric adjusted skew information $I^f(\rho,H)$ w.r.t. $\rho$ (see \cite{MR2430207}), $Q^f(\rho)$ is also convex w.r.t. $\rho$, so the desired inequality (\ref{n-1}) is obtained from that $\rho$ is always a mixture of pure states.
\end{proof}

Now it is desirable to summarize the instructive properties of $Q^f(\rho)$ as a measure of quantum uncertainty:

(a) $Q^f(\rho)$ is unitary invariant, i.e., for any unitary operator $U$, we have $Q^f(U\rho U^{\dag})=Q^f(\rho)$.

(b) $Q^f(\rho)$ is convex w.r.t. $\rho$, which means it decreases when several states are mixed. Since mixing is a classical procedure which can not increase quantum uncertainty.

(c) $Q^f(\rho)$ has the spectrum representation (\ref{spectrum representation}), and has tight bounds (\ref{tight bounds}). Which indicates that the maximally mixed state $\mathbf{1}/n$ has no quantum uncertainty since it commutes with any observable, and all the pure states have the same maximal quantum uncertainty as we expect. 

(d) $Q^f(\rho)$ is monotone with respect to $f$ in the meaning as shown in Corollary \ref{corollary comparison}.

(e) $Q^f(\rho)$ decreases under the partial trace in the following two senses.

First, let $\rho$ be a pure state of composite system $H^a \otimes H^b$, then
\begin{eqnarray}
Q^f(\rho)\geq Q^f(\rho^a).\notag
\end{eqnarray}
Since $Q^f(\rho)$ is equal to $\mathrm{dim}(H^a \otimes H^b)-1$, and $Q^f(\rho^a)$ is less than or equal to $\mathrm{dim}H^a-1$.

Second, let $\{A_j\}$ be an orthonormal base of $L(H^a)$, then
\begin{eqnarray}
\sum_{j=1}^{m^2}I^f(\rho^{ab},A_j\otimes \mathbf{1}^b)-\sum_{j=1}^{m^2}I^f(\rho^a,A_j)\geq 0.
\notag
\end{eqnarray}
The difference between these two quantities (the left side of the above inequality) captures the correlations in $\rho^{ab}$ that can be probed by local observables of $H_a$ (see \cite{PhysRevA.85.032117}). In the next section we will show its application to detect the nonlocality of a composite state $\rho^{ab}$.


\section{Detection of nonlocality and entanglement with $Q^{f}$}\label{secofapplications}

Consider a bipartite state $\rho^{ab}$ of the composite system $H^a \otimes H^b$ with $\mathrm{dim} H^a=m$ and $\mathrm{dim} H^b=n$. To quantify the correlations of $\rho^{ab}$ between systems $a$ and $b$ via $Q^{f}$, we define
\begin{eqnarray}\label{Fbar}
\bar{F}^f(\rho^{ab})=\sum_{j=1}^{m^2}I^f(\rho^{ab},A_j\otimes \mathbf{1}^b)-\sum_{j=1}^{m^2}I^f(\rho^a,A_j),
\end{eqnarray}
and
\begin{eqnarray}\label{Fhat}
\hat{F}^f(\rho^{ab})=\sum_{j=1}^{m^2}I^f(\rho^{ab},A_j\otimes \mathbf{1}^b+\mathbf{1}^a \otimes B_j), \mathrm{\quad when \quad} m=n. 
\end{eqnarray}
Here $\{A_j\}$ and $\{B_j\}$ are orthonormal bases of $L(H^a)$ and $L(H^b)$. Following a similar procedure as in the last section, we can prove these two measures do not depend on the choice of the bases. The measure (\ref{Fbar}) with $f=f^{WY}$ is applied to detect the nonlocality of $\rho^{ab}$ in \cite{PhysRevA.85.032117}. And the measure 
(\ref{Fhat}) with $f(t)=(t+1)/2$, which corresponds to the usual quantum Fisher information (see \cite{Helstrombook}), is applied to detect the entanglement of $\rho^{ab}$ in \cite{PhysRevA.88.014301}. Now we revisit these two results for any regular $f \in \mathcal{F}_{op}$ (that is to say for any metric adjusted skew information).

To show the application of the measure (\ref{Fbar}) in the detection of nonlocality, we need a lemma of the strict convexity of the function $\frac{(t-1)^2}{f(t)}$. The convexity of the function $\frac{(t-1)^2}{f(t)}$ has been well discussed (see \cite{MR2661519,MR2430207}). However the strictness of the convexity has not been clarified as far as I know. So we would like to give the following lemma.

\begin{lemma}\label{strictness of the convexity}
For a regular $f \in \mathcal{F}_{op}$, the convexity of the function $g(t)=\frac{(t-1)^2}{f(t)}$ on $R_+$ is strict. That is to say, if we have $g(\sum_i \lambda_it_i)=\sum_i \lambda_ig(t_i)$, where $0<\lambda_i <1, \sum_{i}\lambda_i=1$ and $t_i \in R_+$, then we get all $t_i$ are the same.
\end{lemma}   
\begin{proof}
In Cai and Hansen$^{\cite{MR2661519}}$, we show that the three functions $\frac{t^2}{f(t)},\frac{-2t}{f(t)}$ and $\frac{1}{f(t)}$ are all convex on $R_+$ (indeed they are operator convex on $R_+$). So it suffices to show one of the three functions, for example, $\frac{1}{f(t)}$ is strictly convex. According to Hansen$^{\cite{MR2430207}}$, $\frac{1}{f(t)}$ has a canonical representation
\begin{eqnarray}
\frac{1}{f(t)}=\int_0^1\Big(\frac{1}{t+\lambda}+\frac{1}{1+t\lambda}\Big)\mathrm{d}\mu(\lambda),\notag
\end{eqnarray}
where $\mu(\lambda)$ is a finite Borel measure on $[0,1]$ and satisfies
\begin{eqnarray}
2\int_0^1 \frac{1}{1+\lambda} \mathrm{d}\mu(\lambda)=1.\notag
\end{eqnarray}
So due to the strict convexity on $R_+$ of $h(t)=\frac{1}{t+\lambda}+\frac{1}{1+t\lambda}$ for all $\lambda \in [0,1]$, $\frac{1}{f(t)}$ is strictly convex on $R_+$.
\end{proof}

\begin{theorem}

\begin{enumerate}
\item[(1)] $\bar{F}^f(\rho^{ab})=0$ if and only if $\rho^{ab}$ is a product state (i.e., $\rho^{ab}=\rho^a \otimes \rho^b$).
\item[(2)] If $\hat{F}^f(\rho^{ab})>2m-2$, then $\rho^{ab}$ must be entangled, that is to say there do not exist $\lambda_j$, $\rho^a_j$ and $\rho^b_j$ such that
\begin{eqnarray}\label{separable}
\rho^{ab}=\sum_j \lambda_j \rho^a_j\otimes \rho^b_j, \quad 0\leq \lambda_j \leq 1 \quad \mathrm{and}\quad \sum_j \lambda_j=1.
\end{eqnarray}
\end{enumerate}

\end{theorem}

\begin{proof}
(1) If $\rho^{ab}=\rho^a \otimes \rho^b$, we will show that $\bar{F}^f(\rho^{ab})=0$. Let $\{\lambda^a_j\}$ be the eigenvalues of $\rho^a$ with eigenvectors $\{|j^a\rangle\}$, and  $\{\lambda^b_k\}$ be the eigenvalues of $\rho^b$ with eigenvectors $\{|k^b\rangle\}$, then
\begin{eqnarray}
\rho^a \otimes \rho^b = \sum_{j,k}\lambda_j^a\lambda_k^b |j^a\rangle\langle j^a| \otimes |k^b\rangle\langle k^b|          . \notag
\end{eqnarray}
For any observable $A$ on system $a$, we have
\begin{eqnarray}
I^f(\rho^a\otimes \rho^b, A\otimes \mathbf{1}^b)&=&-\frac{f(0)}{2}\mathrm{tr} A\otimes \mathbf{1}^b
\frac{(L_{\rho^a\otimes \rho^b}-R_{\rho^a\otimes \rho^b})^2}{m^f(L_{\rho^a\otimes \rho^b},R_{\rho^a\otimes \rho^b})}A\otimes \mathbf{1}^b
\notag\\&=& \frac{f(0)}{2}\sum_{j,k,\hat{j},\hat{k}}\frac{(\lambda_j^a\lambda_k^b-\lambda_{\hat{j}}^a\lambda_{\hat{k}}^b)^2} {m^f(\lambda_j^a\lambda_k^b,\lambda_{\hat{j}}^a\lambda_{\hat{k}}^b)}|\langle j^a|A|\hat{j}^a\rangle|^2\cdot
|\langle k^b|\hat{k}^b\rangle|^2
\notag\\&=& \frac{f(0)}{2}\sum_{j,\hat{j},k}\lambda_k^b \frac{(\lambda_j^a-\lambda_{\hat{j}}^a)^2} {m^f(\lambda_j^a,\lambda_{\hat{j}}^a)}|\langle j^a|A|\hat{j}^a\rangle|^2
\notag\\&=& \frac{f(0)}{2}\sum_{j,\hat{j}}\frac{(\lambda_j^a-\lambda_{\hat{j}}^a)^2} {m^f(\lambda_j^a,\lambda_{\hat{j}}^a)}|\langle j^a|A|\hat{j}^a\rangle|^2 = I^f(\rho^a,A). \notag
\end{eqnarray}
We readily know that  $\bar{F}^f(\rho^{ab})=0$. 

Conversely, if  $\bar{F}^f(\rho^{ab})=0$, then in view of inequality (\ref{weak superadditivity}), we have 
\begin{eqnarray}\label{equality of Fbar}
I^f(\rho^{ab},A\otimes \mathbf{1}^b)=I^f(\rho^a,A)
\end{eqnarray}
for any $A \in L(H^a)$, since we can always, up to a constant normalization, take $A$ as an element of $\{A_j\}$ in the definition of $\bar{F}^f(\rho^{ab})$. Especially for any eigenvector $|j^a\rangle$ of $\rho^a$, we have 
\begin{eqnarray}
I^f(\rho^{ab},|j^a\rangle\langle j^a|\otimes \mathbf{1}^b)=0,\notag
\end{eqnarray}
that is to say
\begin{eqnarray}
K^f(i\big[\rho^{ab},|j^a\rangle\langle j^a|\otimes \mathbf{1}^b\big],i\big[\rho^{ab},|j^a\rangle\langle j^a|\otimes \mathbf{1}^b\big])=0.\notag
\end{eqnarray}
Then we know $\big[\rho^{ab},|j^a\rangle\langle j^a|\otimes \mathbf{1}^b\big]$ must be $0$, which ensures that
\begin{eqnarray}
\rho^{ab}=\sum_j (|j^a\rangle\langle j^a|\otimes \mathbf{1}^b) \rho^{ab} (|j^a\rangle\langle j^a|\otimes \mathbf{1}^b).\notag
\end{eqnarray}
According to Propostion 1 in \cite{MR2407386}, $\rho^{ab}$ must be classical-quantum, i.e.  
\begin{eqnarray}
\rho^{ab}=\sum_{j=1}^{m}p_j|j^a \rangle\langle j^a|\otimes \rho_j.\notag
\end{eqnarray}
Where $\{p_j, j=1,...,m\}$ are the eigenvalues of $\rho^a$ with eigenvectors $\{|j^a\rangle, j=1,...,m\}$. 

Now it  suffices to prove $\rho_j =\rho_k$ for any $j, k =1,...,m$, especially $\rho_1 = \rho_2$. Let $X$ be
\begin{eqnarray}
\frac{\sqrt{2}}{2}|1^a\rangle\langle2^a|+\frac{\sqrt{2}}{2}|2^a\rangle\langle1^a|,\notag
\end{eqnarray}
then
\begin{eqnarray}
I^f(\rho^a,X)=\frac{f(0)}{2}\frac{(p_1-p_2)^2}{m^f(p_1,p_2)}.\notag
\end{eqnarray}
Denote the eigenvalues of $\rho_1$ and $\rho_2$ by $\{\lambda^1_j, j=1,...,n\}$ and $\{\lambda^2_j, j=1,...,n\}$ with eigenvectors $\{|j^1 \rangle , j=1,...,n\}$ and $\{|j^2\rangle, j=1,...,n\}$, then
\begin{eqnarray}
I^f(\rho^{ab}, X\otimes \mathbf{1}^b )&=&\frac{f(0)}{2}\sum_{j,k=1}^{n}\frac{(p_1\lambda_j^1-p_2\lambda^2_k)^2}{m^f(p_1\lambda^1_j,p_2\lambda^2_k)}\big{|}\langle j^1 | k^2 \rangle\big{|}^2
\notag\\
&=&\frac{f(0)}{2}\sum_{j,k=1}^{n}\lambda^1_j\frac{(p_1-p_2\frac{\lambda^2_k}{\lambda^1_j})^2}{m^f(p_1,p_2\frac{\lambda^2_k}{\lambda^1_j})}\big{|}\langle j^1 | k^2 \rangle\big{|}^2
\notag\\
&=&\frac{f(0)\cdot p_1}{2}\sum_{j,k=1}^{n}\frac{(1-\frac{p_2}{p_1}\cdot \frac{\lambda_k^2}{\lambda_j^1})^2}{f(\frac{p_2}{p_1}\cdot \frac{\lambda_k^2}{\lambda_j^1})}\cdot \lambda^1_j\big{|}\langle j^1 | k^2 \rangle\big{|}^2.
\notag
\end{eqnarray}
And
\begin{eqnarray}
I^f(\rho^a, X)&=&\frac{f(0)}{2}\frac{(p_1-p_2)^2}{m^f(p_1,p_2)}=\frac{f(0)\cdot p_1}{2}\frac{(1-\frac{p_2}{p_1})^2}{f(\frac{p_2}{p_1})}
\notag\\
&=&\frac{f(0)\cdot p_1}{2}\frac{\Huge{(}1-\frac{p_2}{p_1}\sum_{j,k=1}^{n}\frac{\lambda^2_k}{\lambda^1_j}\cdot \lambda_j^1\big{|}\langle j^1 | k^2 \rangle\big{|}^2\Huge{)}^2}{f(\frac{p_2}{p_1}\sum_{j,k=1}^{n}\frac{\lambda^2_k}{\lambda^1_j}\cdot \lambda_j^1\big{|}\langle j^1 | k^2 \rangle\big{|}^2)}.\notag
\end{eqnarray}
Combining (\ref{equality of Fbar}) and the strict convexity of $\frac{(t-1)^2}{f(t)}$ ( see Lemma \ref{strictness of the convexity}), we have
\begin{eqnarray}
(\lambda^1_j-\lambda^2_k)\langle j^1|k^2\rangle = 0, \quad j,k=1,...,n,\notag
\end{eqnarray}
which ensure that $\rho_1=\rho_2$. Thus we get $\rho_j=\rho_k$ for any $j,k=1,...,n$ in general. Then $\rho^{ab}$ must be a product state.

(2) If $\rho^{ab}$ is separable, i.e. $\rho^{ab}$ can be formulated as (\ref{separable}), then 
\begin{eqnarray}
\hat{F}^f(\rho^{ab})&=&\sum_{k=1}^{m^2}I^f(\sum_j \lambda_j \rho^a_j\otimes \rho^b_j,A_k\otimes \mathbf{1}^b+\mathbf{1}^a \otimes B_k)
\notag\\ &\leq& \sum_{k=1}^{m^2} \sum_j \lambda_j I^f(\rho^a_j\otimes \rho^b_j,A_k\otimes \mathbf{1}^b+\mathbf{1}^a \otimes B_k)
\notag\\&=& \sum_j \lambda_j \Big( \sum_{k=1}^{m^2}I^f(\rho_j^a, A_k) + \sum_{k=1}^{m^2}I^f(\rho_j^b, B_k)\Big)
\notag\\&=& \sum_j \lambda_j \Big(Q^f(\rho^a)+Q^f(\rho^b)\Big)\leq 2m-2. \notag
\end{eqnarray}
The second inequality follows the convexity of $I^f$, and the last inequality follows Corollary \ref{corollary n-1}.
\end{proof}

Further, prompted by Li and Luo$^{\cite{PhysRevA.88.014301}}$, we would like to combine the entanglement criterion via metric adjusted skew information with the criterion via variance suggested in Hofmann and Takeuchi$^{\cite{PhysRevA.68.032103}}$. If  $\rho^{ab}$ is separable, then the following inequality holds:
\begin{eqnarray}
\sum_{k=1}^{m^2}V( \rho^{ab},A_k\otimes \mathbf{1}^b+\mathbf{1}^a \otimes B_k) \geq 2m-2.\notag
\end{eqnarray}
So the violation of this inequality is also a signature of entanglement. To illustrate how they work, we will evaluate these two criteria for the typical example given in Li and Luo$^{\cite{PhysRevA.88.014301}}$.

\begin{example}
Let $H^a=H^b$ with dimension $m=3$, and $\{|0\rangle,|1\rangle,|2\rangle\}$ be an orthonormal base of $H^a$ (and also of $H^b$). Consider the $3\times 3$ dimensional state
\begin{eqnarray}
\rho^{ab}_p=(1-p)\frac{\mathbf{1}}{9}+p|\Omega\rangle\langle \Omega| \notag
\end{eqnarray}
on $H^a \otimes H^b$ with $\mathbf{1}$ the identity operator on $H^a \otimes H^b$, and 
\begin{eqnarray}
|\Omega\rangle=\frac{1}{\sqrt{3}}\big(|00\rangle+|11\rangle+|22\rangle\big).\notag
\end{eqnarray}
Let
\begin{eqnarray}
&&A_1=|0\rangle\langle0|,\quad A_2=|1\rangle\langle1|, \quad A_3=|2\rangle\langle 2|,
\notag\\ 
&&A_4=\frac{1}{\sqrt{2}}\big(|0\rangle\langle1|+|1\rangle\langle 0|\big),\quad A_5=\frac{1}{\sqrt{2}}\big(|0\rangle\langle 2|+|2\rangle\langle 0|\big),
\notag\\
&&A_6=\frac{1}{\sqrt{2}}\big(|1\rangle\langle2|+|2\rangle\langle 1|\big),\quad A_7=\frac{i}{\sqrt{2}}\big(|0\rangle\langle 1|-|1\rangle\langle 0|\big),
\notag\\
&&A_8=\frac{i}{\sqrt{2}}\big(|0\rangle\langle 2|-|2\rangle\langle 0|\big),\quad A_9=\frac{i}{\sqrt{2}}\big(|1\rangle\langle 2|-|2\rangle\langle 1|\big),\notag
\end{eqnarray}
be a local orthonormal observable base for subsystem $a$ and $\{B_k\}$ be similarly defined for subsystem $b$. Then we obtain
\begin{eqnarray}
\hat{F}^f(\rho^{ab}_p)=\sum_{j=1}^{9}I^f(\rho^{ab}_p,A_j\otimes \mathbf{1}^b+\mathbf{1}^a \otimes B_j)
=\frac{20}{3}f(0)\times\frac{p^2}{m^f(\frac{1}{9}+\frac{8}{9}p,\frac{1}{9}-\frac{1}{9}p)},\notag
\end{eqnarray}
and
\begin{eqnarray}
\hat{V}(\rho^{ab}_p)=\sum_{j=1}^{9}V(\rho^{ab}_p,A_j\otimes \mathbf{1}^b+\mathbf{1}^a \otimes B_j)
=\frac{16}{3}+\frac{4}{3}p.\notag
\end{eqnarray}
Thus when $\hat{F}^f(\rho^{ab}_p)>4$, the state $\rho^{ab}_p$ is entanglement detected by the criterion via metric adjusted skew information. This happens, for example, when $f(t)=(t+1)/2$, and $p=0.7$, then $\hat{F}^f(\rho^{ab}_p)=4.2609$ as calculated in \cite{PhysRevA.88.014301}. But it can not be detected by the criterion via variance, since $\hat{V}(\rho^{ab}_p) \geq 16/3$ for any $0\leq p\leq 1$.
\end{example}



\section{Conclusion and Discussion}\label{discussion}
Based on the essential properties of the metric adjusted skew information, we have investigated the measure of quantum uncertainty $Q^f$ suggested in Gibilisco \cite{gibiliscofisherinf}. These measures can be evaluated by any orthonormal base of the observable space directly, so they are more convenient to quantify the correlations of bipartite state, comparing with some well-known measures of correlations including the entanglement of formation$^{\cite{PhysRevA.54.3824}}$, quantum discord$^{\cite{PhysRevLett.88.017901}}$, the measurement-induced nonlocality$^{\cite{PhysRevLett.106.120401}}$, etc. Although we have exibited some fundamental properties of $Q^f$ such as basis independence and convexity, it is desirable to investigate that if it satisfies more fundamental properties as von Neumann entropy, especially the strong subadditivity and the convexity of the incremental information (see \cite{MR3572497,PhysRevLett.30.434}). And one more interesting question is that, under which assumptions of quantum uncertainty measures, each measure of quantum uncertainty corresponds to a metric adjusted skew information.

\vskip 0.5cm

\section* {Acknowledgements} This work is supported by the National Natural Science Foundation of China (11301025).

\end{document}